\documentclass[13pt,a4paper]{article}

\usepackage{algorithm}
\usepackage{algpseudocode}
\usepackage{amsthm} 
\usepackage{amsmath}
\usepackage{amssymb}
\usepackage{mathtools}

\newcommand{\BigO}[1]{\ensuremath{\operatorname{O}\bigl(#1\bigr)}}

\newtheorem{theorem}{Theorem}[section]

\newtheorem{lemma}[theorem]{Lemma}

\newtheorem{definition}[theorem]{Definition}
\newtheorem{notation}[theorem]{Notation}

\begin{document}


\title{ Sorting Networks: The Final Countdown }
\author{Martin Marinov\\ David Gregg}
\maketitle




\begin{abstract}

In this paper we extend the knowledge on the problem of empirically searching for sorting networks of minimal depth. We present new search space pruning techniques for the last four levels of a candidate sorting network by considering only the output set representation of a network. We present an algorithm for checking whether an $n$-input sorting network of depth $d$ exists by considering the minimal up to permutation and reflection itemsets at each level and using the pruning at the last four levels. We experimentally evaluated this algorithm to find the optimal depth sorting networks for all $n \leq 12$.

\end{abstract}

%

\section{Introduction}

A sorting network is an abstract mathematical model designed to sort numbers in a predetermined sequence of comparators. A sorting network consists of $n$ wires and comparators between pairs of wires such that any input of $n$ numbers is sorted by the network, where one wire corresponds to one number. The two most common measures of sorting networks are the total number of comparators --- \emph{Bose-Nelson's sorting problem~\cite{Bose:1962:SP}} --- and the number of network levels, also referred to as depth. Our work is related to empirically searching for optimal depth sorting networks.

\section{Related Work and Contributions}

Knuth \cite{Knuth73} showed the optimal depth sorting networks for all $n \leq 8$. He also presents the zero-one principle of sorting networks which states that if a comparator network sorts all $2^n$ binary strings of length $n$ then it is a sorting network.

Parberry~\cite{Parberry89} presented a computer assisted proof for the minimal depth of a nine-input sorting network. He significantly reduced network level candidates for the first two levels, in comparison to the naive approach, by exploiting symmetries of the networks (referred to as first and second normal form \cite{Parberry89}). For the remaining network levels he proves that we need to only consider ones with maximal number of comparators. Parberry also found a method, referred to as ``The Heuristic'', to significantly reduce the search space for the second last level of the network. He used a CRAY super computer to test all nine-input comparator networks of depth six that are in second normal form and pass ``The Heuristic'' check. He verified experimentally that none of them are sorting networks. It is an immediate consequence of his result that there does not exist a ten-input sorting network of depth six. The presented pruning techniques in this paper are at least as good as Parberry's ones for the last two levels, and the third and fourth last levels our work is novel.

Codish~\cite{CodishCS14a_The_End_Game} presents safe pruning techniques for the last layer of a sorting network that are aimed at improving algorithms that use the SAT encoding of sorting networks and transform the problem to SAT problem. Parberry~\cite{Parberry89} has already presented an 'on-the-fly' method of constructing the last layer. Codish's work is related to the case when the ``current'' comparator network is not ``known'' (by the algorithm). Hence they invented conditions that would suit this specific case of encoding the optimal depth sorting network problem as a SAT problem.  If the current comparator network (or its full output set) is known by the algorithm then Parberry's result is much stronger than Codish's. In this paper, we focus on the case when the network is known and hence we present improvements over Parberry's technique.

Bundala~\cite{BundalaCCSZ14_Optimal_Depth} presented a computer assisted proof for the optimal depths of networks with eleven to sixteen (inclusive) inputs. He also managed to significantly reduce the number of candidates for the second layer in comparison to Parberry's approach, by considering only networks whose outputs are minimal representative up to permutation and reflection. Similar work for the second level is also presented by Michael Codish in \cite{CodishCS14:Two_Layer_Prefix}. Bundala's algorithm for finding sorting networks of optimal depth is based on a SAT encoding of the optimal depth sorting networks problem, which uses the set of candidate two-layer networks as a fixed entry point. Some extra pruning techniques are presented and they use a state of the art SAT solver to find the optimal depth sorting networks for all $n \leq 16$. 

Our algorithm and all related techniques are developed independently of that of Bundala \cite{BundalaCCSZ14_Optimal_Depth} and Codish \cite{CodishCS14:Two_Layer_Prefix} \cite{CodishCS14a_The_End_Game}. Using our program, we manage to prove the optimal sorting networks for all $n \leq 12$. Although, the approach presented in this paper significantly differs from that of Bundala; instead of using a generic SAT solver, our method studies details about the structure of sorting networks and gives much more enriched answer to the yes/no question being asked ('does there exists and $n$-input sorting network of depth $d$?'). We give more insight on the candidates that are needed to be considered at each level --- also referred to as complete set of filters $F_{n, d}$ for comparator networks of $d$ \cite{Marinov:SortingNetworks:ThirdLevel}.

Marinov~\cite{Marinov:ExtremalSets:Permutation} presented a highly efficient practical algorithm for finding the minimal representative itemsets over a domain $D$ up to a permutation of $D$. This algorithm can be applied to reduce the number of candidates for the second layer as described in Bundala and Codish, although \cite{BundalaCCSZ14_Optimal_Depth} and \cite{CodishCS14:Two_Layer_Prefix} present an extra pruning method using reflection.

Marinov~\cite{Marinov:SortingNetworks:ThirdLevel} presented a modified version of \cite{Marinov:ExtremalSets:Permutation} that finds the comparator networks whose outputs are minimal representative up to permutation and reflection $R_{n, 3}$ for the first three layers. This significantly reduces the search space size for any $n \geq 3$. This technique can be easily adapted to find the sets $R_{n, d}$ for and depth $d$, which means that we can further reduce the size of the search space using Marinov's existing technique. Marinov's technique is also applicable to the SAT encoding of a network and would also speedup up Bundala's algorithm by fixing the first three layers of a network. rather than just the first two as described in \cite{BundalaCCSZ14_Optimal_Depth}.

\subsection{ Problem Statement }

The problem addressed in this paper is that of reducing the candidate networks of any depth that need to be considered, i.e. for every $d$ we want to find a set $X_d \in F_{n, d}$ that is minimal in size. Solving this problem inevitably leads us to a new algorithm for finding sorting networks of minimal depth which was first detailed by Knuth \cite{Knuth73} more than 40 years ago. If we ask the question ``Does there exists an $n$-input sorting network of depth $d$'' then our work is about reducing the sizes of the sets $X_{d-3}$, $X_{d-2}$, $X_{d-1}$ and $X_{d}$

\subsection{Contributions}

\begin{itemize}
	\item \emph{Construct the From, To and Reach sets for the last $k$ levels a comparator network using output sets} --- we present a novel method to construct the From, To and Reach sets using outputs sets instead of comparator networks as described by Parberry~\cite{Parberry89}. For the case of $k = 2$ we prove that for any channel $c$ the $From_P(c)$, $To_P(c)$ and $Reach_P(c)$ sets computed by Parberry's approach are contained in the respective sets computed by our new method . Hence we prove that when $k = 2$ our new approach rejects at least as many candidates as Parberry's method - ``The Heuristic'' \cite{Parberry89}. 
	
	\item \emph{New theory for finding comparator networks that cannot be extended to sorting ones by three or four levels} --- we present a novel theoretical upper bound on the sizes of the From, To and Reach sets of any $n$-input comparator network for third and fourth last levels. We apply these upper bounds to detect comparator networks that cannot be extended to sorting ones with the addition of three or four levels. Hence, we further reduce the search space when checking if an $n$-input sorting network of depth $d$ exists by rejecting comparator network candidates of depth $d-4$, $d-3$ when checking if an $n$-input sorting network of depth $d$ exists.
	
	\item \emph{Algorithm for finding minimal depth $n$-input sorting networks } --- our new algorithm presented in this paper is a combination of finding the minimal representative up to permutation and reflection itemsets at every depth and also applying our new sortable in two, three and four level search space reduction techniques. This algorithm gives transparency on the number of candidate networks considered at each depth which is highly different than the current state of the approach.
	
	\item \emph{Computer assisted proof that the minimal depths $n$-input sorting network for all $n \leq 12$} --- at the time when our work was developed, it was previously thought that sorting networks of depth seven could exist for eleven and twelve-input networks but had not yet been found. The optimality of sorting networks for $n \leq 12$ that is presented in this paper is not to be considered as novel because Bundala~\cite{BundalaCCSZ14_Optimal_Depth} has already proven it.
\end{itemize}

\section{ Background }

\subsection{ Formal Definition of Comparator/Sorting Networks }

\begin{definition}
\label{ComparatorNetworkDefinition}
A \emph{generalized comparator} is an ordered pair $\langle i, j \rangle$ such that $1 \leq i \neq j \leq n$. A generalized comparator is a \emph{comparator} or \emph{min-max comparator} if $i < j$. The values $i$ and $j$ are referred to as \emph{channels}.
A \emph{generalized level} $L$ is a set of generalized comparators such that each channel is involved in at most one generalized comparator, formally if $\langle a, b \rangle, \langle c, d \rangle \in L$ then $\lvert \{ a, b, c, d \} \lvert = 4 $. A generalized level is a \emph{level} or \emph{min-max level} if it consists only of (min-max) comparators. The set of all (min-max) levels is denoted as $G_n$, as described by Bundala~\cite{BundalaCCSZ14_Optimal_Depth}.
A \emph{generalized $n$-input comparator network} is a vector $\langle L_1, L_2, \dots, L_d, n \rangle $, where $L_1, L_2, \dots, L_d$ are generalized levels, and $n$ is a positive integer. A generalized $n$-input comparator network is called an \emph{$n$-input comparator network} if it consists only of (min-max) levels. Let $C = \langle L_1, L_2, \dots, L_d, n \rangle$ be an $n$-input comparator network, we define the size of $C$ as the number of levels, i.e. $|C| = d$.
\end{definition}

So far we have formally defined the structure of a (generalized) comparator network. We need to define the \emph{output} of applying a comparator network to an \emph{input}, where an input is an n-bit binary string \cite{Knuth73}. Applying a network to an input permutes the input vector. Hence, for any fixed input we can define a permutation that models the network behaviour when applied to that particular input.

\begin{notation}
\label{PermutationDefinition}
Denote the set of all permutations of n elements as $\Pi_n = \{ \pi : \{1, 2, \dots, n\} \longmapsto \{1, 2, \dots, n\}$ $\lvert$ $\pi$ is bijective $\}$. Let $v = \langle a_1, a_2, \dots \rangle$ be a vector. Denote by $v_i$ the $i$-th coordinate of $v$, namely $v_i = a_i$.
\end{notation}

\begin{definition}
\label{def:ComparatorNetworkEvaluation}
An \emph{input} is a vector $\mathbf{x} \in \{0,1\}^n$ as per Knuth's~\cite{Knuth73} zero-one principle. Denote by $I_n$ the set of all inputs. The evaluation of a generalized $n$-input comparator network $C = \langle L_1, \dots, L_d, n \rangle $ in channel $i$ at level $k$ on input $x$ is the two dimensional vector $e_x(i, k)$ where:

\[
 e_x(i, k) =
  \begin{cases} 
	\langle x_i, i \rangle		& \text{ $ k = 0 $}											\\
    e_x(i, k-1)		& \text{ $ \langle i, j \rangle \in L_k $ and $ e_x(i, k-1)_1 <= e_x(j, k-1)_1 $}	\\
    e_x(j, k-1)		& \text{ $ \langle i, j \rangle \in L_k $ and $ e_x(i, k-1)_1 >  e_x(j, k-1)_1  $}	\\
    e_x(i, k-1)		& \text{ $ \langle j, i \rangle \in L_k $ and $ e_x(i, k-1)_1 >= e_x(j, k-1)_1 $}	\\
    e_x(j, k-1)		& \text{ $ \langle j, i \rangle \in L_k $ and $ e_x(i, k-1)_1 <  e_x(j, k-1)_1$}	\\
	e_x(i, k-1)		& \text{ otherwise }										\\
  \end{cases}
\]

The output of applying $C$ to $x$ is \emph{$V_C(x) = \langle e_x(1, d)_1, \dots, e_x(n, d)_1 \rangle \in I_n$}. The permutation of the coordinates when applying $C$ to  $x$ is \emph{$P_C(x) = \langle e_x(1, d)_2, \dots, e_x(n, d)_2 \rangle \in \Pi_n $}.
\end{definition}

Intuitively, we say that a vector in $I_n$ is sorted if its values are non-decreasing left-to-right, and a sorting network is one which sorts all possible $2^n$ input vectors. More formally:

\begin{definition}
\label{SortingNetworkDefinition}
The vector $\langle x_1, x_2, \dots, x_n \rangle \in \{0,1\}^n$ is \emph{sorted} iff $x_i <= x_{i + 1}$ for all $1 \leq i < n$.
A \emph{generalized sorting network} is a generalized $n$-input comparator network for which there exists a permutation $\pi \in \Pi_n$ such that $\pi( V_C(x) )$ is sorted for all inputs $ x \in I_n$.
A \emph{sorting network} is an $n$-input comparator network such that $ V_C( x ) $ is sorted for all inputs $ x \in I_n$.
\end{definition}

\begin{theorem}
\label{FloydKnuthGeneralizedTheorem}
For every generalized sorting network there is a sorting network with the same size and depth. If the former has only min-max comparators in the first k levels, then the latter is identical in the first k levels.
\end{theorem}

\begin{proof}
See Knuth~\cite{Knuth73}.
\end{proof}

\subsection{Known Comparator Network Theory}
\label{sec:background:theory:overlap}

\begin{definition}
\label{OutputSetDefinition}
Let the \emph{output set} of a comparator network $C$ be $S_C = \{ V_C(x) \lvert x \in I_n \}$. Let the set of all already sorted inputs $T_n = \{ (x_1, x_2, \dots, x_n)$ $\lvert$  $x_{i < j} = 0, x_{i >= j} = 1$ for $ 1 \leq j \leq n + 1 \}$.
\end{definition}

\begin{definition}
\label{ComparatorNetworkUnionLevelDefinition}
Let $A$ and $B$ be $n$-input comparator networks, where $A = \langle A_1, A_2, \dots, A_d, n \rangle $, $B = \langle B_1, B_2, \dots, B_k, n \rangle $, and let $L$ be a level. Define the concatenations $A \oplus L = \langle A_1, \dots, A_d, L, n \rangle $ and $A \oplus B = A \oplus B_1 \oplus B_2 \oplus \dots \oplus B_k$. Note that $\oplus$ is associative.
\end{definition}

\begin{theorem}
\label{MinimalOutputSetTheorem}
Let $A$, $B$ and $C$ be $n$-input comparator networks. Suppose that $S_{A} \subseteq S_{B}$ and $B \oplus C$ is an $n$-input sorting network. Then there exists a comparator network $C'$ with the same depth as $C$ such that $A \oplus C'$ is an $n$-input sorting network.
\end{theorem}

\begin{proof}
See proof of Theorem~3.8 in \cite{Marinov:SortingNetworks:ThirdLevel}.
\end{proof}

Knuth~\cite{Knuth73} has shown that comparator networks are just as powerful as generalized comparator networks. He shows that the group of generalized comparator networks is closed under permutation. Intuitively, we would like to strengthen the result of Theorem~\ref{MinimalOutputSetTheorem} by considering permutations of output sets. Before we present this result, we need the following lemma to prove it.

\begin{lemma}
\label{lemma:PermutationSet}
Let $\pi \in \Pi_n$, $x \in I_n$, and $C$ be a comparator network such that $V_C(\pi(x))$ is sorted. Then $\pi(V_{\pi^{-1}(C)}(x))$ is sorted, where $\pi^{-1}(C)$ is a generalized comparator network.
\end{lemma}

\begin{proof}
See proof of Lemma~3.9 in \cite{Marinov:SortingNetworks:ThirdLevel}.
\end{proof}

Theorem~\ref{MinimalOutputSetTheorem} tells us that if we can extend the comparator network $B$ to a sorting network by appending $l$ levels to it then we can extend any network $A$ such that $S_A \subseteq S_B$ by appending $l$ levels to it. We now extend this result by weakening the constraint $S_A \subseteq S_B$. We show that it is enough to find one permutation $\pi \in \Pi_n$ such that $\pi(S_A) \subseteq S_B$ to claim that if we can extend the comparator network $B$ to a sorting network by appending $l$ levels to it then we can extend any network $A$ by appending $l$ levels to it.

\begin{theorem}
\label{MinimalPermutationOutputSetTheorem}
Let $A$, $B$ and $C$ be $n$-input comparator networks, and $\pi \in \Pi_n$ such that $ \pi(S_{A}) \subseteq S_{B}$ and $B \oplus C$ is an $n$-input sorting network. Then there exists a comparator network $C'$ with the same depth as $C$ such that $A \oplus C'$ is an $n$-input sorting network.
\end{theorem}

\begin{proof}
See proof of Theorem~3.10 in \cite{Marinov:SortingNetworks:ThirdLevel}.
\end{proof}

\begin{definition}
\label{def:min_pi}
Let $X$ be a set of output sets of $n$-input comparator networks. Define the set of all minimal representative output sets up to permutation of $X$ as $MinPi ( X ) = \{ S_A ~\lvert~ S_A \in X ~:~\nexists~ S_B \in X, \pi \in \Pi_n : B < A, \pi(S_B) \subseteq S_A  \}$, where by $B < A$ we denote the lexicographic order of networks, as described by Parberry~\cite{Parberry89}. Let the set of all output sets of $n$-input comparator networks of depth $d$ be defined as $G_{n,d}$. Let the set of all minimal representative output sets of $n$-input comparator networks of depth $d$ up to permutation be defined as $S_{n,d} = MinPi(G_{n,d})$.
\end{definition}

\begin{definition}
\label{def:filters}
The set $X_n$ of $n$-input comparator networks is a complete set of filters iff for any $n$-input sorting network of depth $d$ there exists one of the form $C : C'$ of depth $d$ for some $C \in X_n$. We would also denote the set of all complete sets of filters of that contain only $n$-input comparator networks with exactly $i$ levels as $F_{n, i} = \{ X_n ~|~ X_n ~is~a~complete~set~of~filters~and~ C \in X_n \implies |C| = i \}$ .
\end{definition}

\begin{definition}
\label{def:network:reflect}
Let $ x = \langle x_1, x_2, \dots, x_n \rangle \in I_n $ then $ \overline{x^R} = \langle \overline{x_n}, \overline{x_{n-1}}, \dots, \overline{x_1} \rangle $ where $x_i \in \{0, 1\}$ and $\overline{0} = 1$ and $\overline{1} = 0$.
Let $ L $ be a level the its reflection $L^R = \{ \langle n - j + 1, n - i + 1 \rangle  ~|~ \langle i, j \rangle \}$.
Let $ C = \langle L_1, L_2, \dots, L_d, n \rangle $ be a comparator network then its reflection $C^R = \langle L_1^R, L_2^R, \dots, L_d^R, n \rangle $.
\end{definition}

\begin{lemma}
\label{lemma:network:reflect}
Let $ C $ be a comparator network then $x \in S_C \iff \overline{x^R} \in S_{C^R} $.
\end{lemma}

\begin{proof}
Refer to the proof of Lemma~8 in \cite{BundalaCCSZ14_Optimal_Depth} by Michael Codish.
\end{proof}

\begin{lemma}
\label{lemma:reflect}
Let $R_{n, i}$ be the set of minimal representative up to permutation and reflection itemsets within $G_{n}$. Then $R_{n, i} \in F_{n, i}$.
\end{lemma}

\begin{proof}
Refer to section~4.2 in \cite{BundalaCCSZ14_Optimal_Depth}.
\end{proof}

\subsection{Parberry's~\cite{Parberry89} Theory for Pruning Last Two Layers}

Parberry devised a pruning technique (``The Heuristic''~\cite{Parberry89}) to reject network level candidates for the second last level when checking whether an $n$-input sorting network of depth $d$ exists. His technique is based on tracking the individual values at channels and giving upper bounds on the number of distinct places a value at a fixed channel $c$ could get send in a two level comparator network. We first give the definition of the sets and then present the upper bounds, as described by Parberry~\cite{Parberry89}. 

\begin{definition}
\label{def:connect}
Let $\langle L_1, \dots, L_d, n \rangle $ be a comparator network of depth $d$, $c$ be a channel and $P \subseteq \{1, 2, \dots, n \}$ then for every $i \geq 1$ define the functions:
\begin{itemize}
\item $Coneect_i(P) = P \cup \{ q$ $\lvert$ $\exists$ $p \in P : \langle p, q \rangle \in L_i$ or $\langle q, p \rangle \in L_i \}$
\item $To_C(c, i) = \{c\} \cup Connect_{i}( Connect_{i+1}( \dots ( Connect_d ( \{c\} ) ) \dots ) )$
\item $From_C(c, i) = \{c\} \cup Connect_{d}( Connect_{d - 1}( \dots ( Connect_i ( \{c\} ) ) \dots ) )$
\item $Reach_C(c, i) = \{c\} \cup To_C(c, i) \cup From_C(c, i)$
\end{itemize}
\end{definition}

\begin{lemma}
\label{lem:reach_2}
Let $C$ be an $n$-input sorting network of depth $d \geq 2$. Then for any channel $c$ the following hold $\lvert To_C(c, d - 2) \lvert \leq 4$, $\lvert From_C(c, d - 2) \lvert \leq 4$ and $\lvert Reach_C(c, d - 2) \lvert \leq 5$.
\end{lemma}

\begin{proof}
Refer to Parberry \cite{Parberry89}.
\end{proof}

\begin{definition}
Let the graph of a comparator network $C$ be defined as $G(C) = (V, E)$, where \newline
$V = \{ (i,j)$ $\lvert$ $1 \leq i \leq n$ and $0 \leq j \leq d \}$ \newline
$E = \{ ((i_1,j_1),(i_2,j_2))$ $\lvert$ $i_1 + 1 = i_2$ and either $j_1 = j_2$ or $ \langle j_1,j_2 \rangle \in L_{i_2} $ or $ \langle j_2,j_1 \rangle \in L_{i_2} $ $ \}$.
\end{definition}

The following lemma is used by Parberry's approach to calculate the from, to and reach sets in practice. We also use it as a reference point to the proves of our new theory presented in the next section.

\begin{lemma}
\label{lem:SortableInTwoLevelsParberry}
Let $A$ and $B$ be comparator networks such that $C = A \oplus B$ is a sorting network. Let $x, y \in I_n$ differ only in the i-th dimension, where $x = (\dots, a_i = 1, \dots)$ and $y = (\dots, x_i = 0, \dots)$. Let $d = \lvert A \lvert$ and $(0, j_0), (1, j_1), \dots, (d, j_d)$ be the sequence of labels of vertex where $G(A, x)$ and $G(A, y)$ differ. Then channel $||x|| \in From_C(j_d, d+1)$ and channel $j_d \in To_C(||x||, d+1)$.
\end{lemma}

\begin{proof}
Refer to Parberry \cite{Parberry89}.
\end{proof}

\section{New Theory for Pruning Last Four Levels}
\label{sec:new:theory}

We present a our new algorithm to construct the from, to and reach sets which is at least as good as Parberry's algorithm. Meaning, if Parberry's method rejects a network candidate then so does ours but the converse is not true in the general case. Hence, we are able to reject more candidates than Parberry's method.

Also, we present and prove upper bounds for the sizes of the from, to and reach sets for the third and fourth last levels which allow us to further reduce the search space significantly in comparison to Parberry's and Bundala's methods as all existing algorithms considers all possible network level candidates for the third and fourth last levels.

\subsection{Constructing the $From$, $To$ and $Reach$ Sets}
\label{sec:construct_reach}

Our aim is to find methods of constructing the $From_C$, $To_C$ and $Reach_C$ sets of any comparator network $C$ that would be superior to that of Parberry's. The following results gives us a way to practically solve this task by find elements in the $From_C$, $To_C$ and $Reach_C$ sets of any comparator network $C$ by considering only the output set representation $S_C$ of $C$. 

\begin{lemma}
\label{lem:OutputSetsDiff}
Let $C$ be comparator network. Let $x, y \in S_C$ be such that they differ in exactly one dimension. Then there exist inputs $v, w \in I_n$ which differ in exactly one dimension such that $V_C(v) = x$ and $V_C(w) = y$. 
\end{lemma}

\begin{proof}
The result follows by induction on the depth $d$ of $C$.

\textbf{Base:}
$d=0$. Then Lemma~\ref{lem:OutputSetsDiff} holds for $w = x$ and $w' = y$.

\textbf{Assumption:}
Lemma~\ref{lem:OutputSetsDiff} holds for all $C$ of depth smaller then $d$.

\textbf{Induction Step:}
Let $C = (L_1, L_2, \dots, L_{d-1}, L_d)$. There are three cases to consider.

\indent
\indent
\underline{Case 1.} $ \langle i, k \rangle, \langle k, i \rangle \notin L_d$ for $1 \leq k \leq n$.
\indent
Follows from assumption because the value of channel $i$ at level $d-1$ is unchanged at level $d$ for any input, in particular $w$ and $w'$.
\newline
\indent
\indent
\underline{Case 2.} $ \langle k, i \rangle \in L_d$.
\indent
Then $v_x(d-1,k) = 0$ and $v_x(d-1, i) = 0$ because $v_x(d, i) = 0$. 
Let $s = (\dots,x_k = 1,\dots,x_i=0,\dots)$ and $t = (\dots,x_k = 0,\dots,x_i=0,\dots)$ be such that after applying $L_d$ to $s$ and $t$ we get $v$ and $v'$ respectively. Using the assumption we apply Lemma~\ref{lem:OutputSetsDiff} to $s$, $t$ and $(L_1,\dots, L_{d-1})$ to get the desired result.
\newline
\indent
\indent
\underline{Case 2.} $ \langle i, k \rangle \in L_d$.
\indent
Then $v_y(d-1,k) = 1$ and $v_y(d-1, i) = 1$ because $v_y(d, i) = 1$. Proof is analogous to Case 2 by setting $t = (\dots,x_k = 1,\dots,x_i=1,\dots)$.

\end{proof}

\begin{theorem}
\label{th:construct_reach}
Let $A$ and $B$ be comparator networks such that $C = A \oplus B$ is a sorting network. Let $d = \lvert A \lvert$ and $x, y \in S_A$ differ only in the i-th dimension, where $x = (\dots, a_i = 1, \dots)$ and $y = (\dots, a_i = 0, \dots)$.  Then $||x|| \in To_C(i, d+1)$ and $i \in From_C(||x||, d+1)$.
\end{theorem}

\begin{proof}
We apply Lemma~\ref{lem:OutputSetsDiff} to $x$ and $y$ to find inputs $v, w \in I_n$ which differ in exactly one dimension such that $V_C(v) = x$ and $V_C(w) = y$. Applying Lemma~\ref{lem:SortableInTwoLevelsParberry} to $v$ and $w$ yields that $||x|| \in To_C(i, d+1)$ and $i \in From_C(||x||, d+1)$.
\end{proof}

Assume that we are given the output set $S_A$ of any comparator network $A$. The above Theorem~\ref{th:construct_reach} gives us a way to construct the sets $from$, $to$ and $reach$ which must be contained in $From_{A \oplus B}$, $To_{A \oplus B}$ and $Reach_{A \oplus B}$ sets respectively for any comparator network $B$ such that $A \oplus B$ is a sorting network by only considering outputs $S_A$ that differ in exactly one dimension. The pseudo code for constructing the $from$, $to$ and $reach$ sets is presented in Algorithm~\ref{algo:construct_reach}, the correctness of which is given by Theorem~\ref{th:construct_reach}.

\subsection{Second Last Level}
\label{sec:sortable_two}

We know that if we are given an output set $S_A$ of any comparator network $A$ we can construct the minimal $from$, $to$ and $reach$ sets for extending $A$ to a sorting network. Using the cardinality upper bounds for any comparator network from Lemma~\ref{lem:SortableInTwoLevelsParberry} we can devise a safety check for whether $A$ can be extended to a sorting network with the addition of two levels. Hence, we can find all network levels candidates for the second last level when trying to extend the network $A$ to an $n$-input sorting network by adding two levels. The pseudo code for finding the second last level candidates given an the output set $S_A$ is presented in Algorithm~\ref{algo:sortable_two}.

\subsection{The Third Last Level}

The next lemma extends the theory of checking whether a comparator network can be sorted with the addition of two levels. We give an upper bounds on the sizes of the $To_C$ and $From_C$ sets for the last three levels of any comparator network $C$ at any channel. Using these bounds we design a safety check to reject comparator networks that are not be extendible to sorting networks with the addition of three levels, similar to the Method presented in the previous section~\ref{sec:sortable_two}

\begin{lemma}
\label{lem:reach_3}
Let $C$ be an $n$-input sorting network of depth $d > 2$. Then for any channel $c$ the following hold $\lvert To_C(c, d - 3) \lvert \leq 8$, $\lvert From_C(c, d - 3) \lvert \leq 8$.
\end{lemma}

\begin{proof}
Follows immediately from the proof of Lemma~\ref{lem:reach_2} by considering all possibilities when we add an extra level.
\end{proof}

\subsection{The Look-Ahead}

Out next result --- referred to as ``The Look-Ahead'' --- tells us that if we try extend the comparator network $A$ by exactly one level then the size of the set of distinct $from$ sets of any channel $c$ is at most $n$. In other words, if we extend $A$ by one level then the $from$ set of a channel $c$ is the same for all levels that contain the comparator $\lvert c, q \rvert$ for a fixed channel $q$. Meaning that if the level $L$ that we extend $A$ by contains the channel $\langle c, q \rangle$ then the other comparators of this level have no effect on the $from$ set at channel $c$, and channel $q$.

\begin{lemma}
\label{lem:look_ahead}
Let $A$ and $B$ be $n$-input comparator networks, where the depth of $A$ is $d$. Let $L$ and $L'$ be levels such that for a fixed channel $c$ we have $\{ q$ $\lvert$ $\langle c, q \rangle \in L_d \} = \{ q'$ $\lvert$ $\langle c, q' \rangle \in {L_d}' \}$. Then $Form_{A \oplus L \oplus B}(c. d+1) = From_{A \oplus L' \oplus B}(c, d+1)$.
\end{lemma}

\begin{proof}
Follows immediately from Definition~\ref{def:connect}. 
\end{proof}

Lemma~\ref{lem:look_ahead} allows us to reject any forth last network level candidate that contains a comparator $\langle p, q \rangle$ for which $| from(p) | > 8$ or $| from(q) | > 8$ for the comparator network $A \oplus \{ \langle p,q \rangle \}$, where the $from$ sets are computed as described in section~\ref{sec:construct_reach}. The pseudo for pruning networks that cannot be extended to sorting ones by the addition of four levels is presented in Algorithm~\ref{algo:sortable_four}.

Using Lemma~\ref{lem:look_ahead}, we can derive similar argument for the look-ahead at level two by checking $| from(p) | > 4$ or $| from(q) | > 4$ for any comparator $\langle p,q \rangle$ that is part of the third last layer. Although, we do not present pseudo code for this method (trivial task) we have encoded it in the implementation of the algorithm presented in section~\ref{sec:algorithm}.

\section{Algorithm for Finding $n$-input Sorting Networks of Minimal Depth}
\label{sec:algorithm}

The pseudo code of the algorithm for checking whether an $n$-input sorting network of depth $d$ exists is presented in Algorithm~\ref{algo:MinDepthSortNet}. Altogether, it is a direct implementation of the theory presented in section~\ref{sec:new:theory} together with the highly efficient algorithm \cite{Marinov:SortingNetworks:ThirdLevel} for finding minimal representative itemsets up to permutation and reflection. The algorithm computes the output sets that are minimal representative up to permutation and reflection at each level $t$, $R_{n, t}$. If $d - t \leq 4$, we remove elements from this set if they cannot be extended to sorting networks with the addition of $d- t$ levels using the already described safety pruning checks from section~\ref{sec:new:theory}.

\subsection{Algorithm Correctness}
\label{sec:algorithm:correctness}
In this section we give a few results that are used to prove the correctness of Algorithm~\ref{algo:MinDepthSortNet} which is given in terms of comments in the pseudo code that point to proven theoretical results.

Definition~\ref{OutputSetDefinition} gives us a map $M : C \mapsto S_C$ from a comparator network, defined as levels of comparators, as per Definition~\ref{ComparatorNetworkDefinition}, to an output set. $M$ is not injective, which implies that the number of output sets of comparator networks of depth $d$ is bounded above by the number of comparator networks of depth $d$. As our goal is to find minimal depth $n$-input sorting networks it is important to present a method for checking if a comparator network $C$ is a sorting network by only considering its output set $S_C$.

\begin{theorem}
\label{SortinNetworkOutputSetSizeTheorem}
An $n$-input comparator network $C$ is a sorting network iff $|S_C| = n + 1$.
\end{theorem}

\begin{proof}
If $C$ is a sorting network then $S_C$ must be equal to $T_n$ because every input is sorted by $C$, hence $|S_C| = n + 1$. It is obvious that $T_n \subseteq S_C$ because $V_C(t) = t$ for any $t \in T_n$. If $|S_C| = n + 1$ then we know that $S_C = T_n$ because $|T_n| = n + 1$ and $T_n \subseteq S_C$, hence $C$ is a sorting network.
\end{proof}

Looking at Algorithm~\ref{algo:MinDepthSortNet}, given the set $R_{n,d}$ we need to be able to answer the question of whether there exists an $n$-input sorting network of depth $d$. As discussed, we can apply Theorem~\ref{SortinNetworkOutputSetSizeTheorem} to every element in the set of output sets $R_{n,d}$ but the following lemma gives us a more practically useful result, i.e. computationally cheaper to check.

\begin{lemma}
\label{AlgoCorrectExistanceLemma}
Suppose that there exists an $n$-input sorting network of depth $d$ then $ \lvert R_{n,d} \lvert = 1 $.
\end{lemma}

\begin{proof}
Let $C$ be an $n$-input sorting network of depth $d$ and let $B$ be an $n$-input comparator network. Since $R_{n,d}$ contains the minimal representative up to permutation and reflection outputs of itemsets and $T_n \in R_{n, d}$, we have $S_C = T_n$. Since $T_n \subseteq S_B$, then $S_C \subseteq S_B$. Therefore $R_{n,d} = \{ S_C \} = \{ T_n \}$, hence $ \lvert R_{n,d} \lvert = 1 $.
\end{proof}

\section{Empirical Evaluation}
\label{sec:experiments}

Clearly, in terms of the presented algorithm, the optimal depth of an $n$-input sorting network is the smallest $d$ for which Exists-Sorting-Network$(n,d)$ returns $true$. Now we move onto the experimental evaluation of this algorithm and compare it to the current state of the art approach.

\subsection{Environment Setup}
In all of the conducted experiments we used a computer with four Intel Xeon CPU E7- 4820 processors. Each CPU has 8 cores clocked at 2.00GHz, equipped with 8MB of third level cache and 128GB of main memory.

We have summarized a subset of the conducted experiments in Figure~\ref{fig:experiments}. We will now describe in detail how to interpret the data presented in this table by focusing on the row for $n = 11$. All other rows can be interpreted in the same logical manner. The value for $|R_{n,3}|$ is the total number of networks of depth three that are a complete set of filters \cite{Marinov:SortingNetworks:ThirdLevel}. To everyone of these itemsets our algorithm tries to apply every level from $G_n$ and the look-ahead pruning check manages to reduce this number from $|R_{n,3}| * |G_n| = 361\,564\,784$ to $105\,290\,955$ itemsets that could be sortable in three levels. We then find the itemsets that can be sortable in three levels by applying Lemma~\ref{lem:reach_3} to find that only $866\,314$ networks could be extended to sorting ones with the addition of three levels. Similarly, our algorithm then applies all network levels to the sortable in three level itemsets to find the networks that could be sortable in two levels; in this case no such networks exist. Since our pruning techniques are safe we can deduce that there does not exist an eleven-input sorting network of depth seven because there are no networks of depth five that could be extended to sorting ones with the addition of two levels. The total runtime of $311$ seconds is in the same order of magnitude as the total runtime of $105$ seconds (in \cite{BundalaCCSZ14_Optimal_Depth} Table~8 the sum of values for row $n = 11$ at columns ``BEE'' + ``SAT'') that it took Bundala's program to prove that there does not exist an eleven-input sorting network of depth seven; although, it is important to note that we have used different machines (architecture, RAM, CPU power, $\dots$) and that no real conclusion could be drown from comparing these reported runtimes.

\begin{figure} [H]
\centering
\resizebox{\linewidth}{!}{%
\begin{tabular}{|c|c|c|c|c|c|c|c|c|c|c|c|c|c|c|c|c|c|c|c|}
\hline
$n$ 	& $|R_{n,3}|$ & Third-Look-Ahead & Sortable-In-Three & Second-Look-Ahead & Sortable-In-Two & Runtime  \\ \hline
$11$	& $10\,129$ & $105\,290\,955$ & $866\,314$ & $0$ & $0$ & $311$ \\ \hline
$12$	& $117\,517$ & $72\,319\,916$ & $7\,847$ & $0$ & $0$ & $210$ \\ \hline
\end{tabular}
}
\caption{Experimental evaluation of Algorithm~\ref{algo:MinDepthSortNet} for checking whether there exists a sorting network of depth seven for eleven and twelve inputs. The column $|R_{n, 3}|$ gives the number of minimal representative up to permutation and reflection outputs of comparator networks of depth three. The column Third-Look-Ahead presents the total number of output sets that could be sortable in three levels as per Algorithm~\ref{algo:sortable_four}. Sortable-In-Three is a reduction on the column Third-Look-Ahead without applying any levels as per Algorithm~\ref{algo:sortable_three}. Second-Look-Ahead gives the total number of output sets that could be sortable in two levels. Sortable-In-Two is the total number of output sets that are sortable in two levels as per Algorithm~\ref{algo:sortable_two}. The Runtime column presents the total execution time in seconds; i.e. if program is executed on two threads and the individual CPU time of each thread is $5$ seconds then the total Runtime would be $10$ seconds. }
\label{fig:experiments}
\end{figure}

\section{Conclusion and Future Work}

In this paper, we have presented new pruning techniques for the last four levels of a sorting network. These techniques are superior to the ones presented by Parberry (developed only for last two levels). We do not compare our techniques to the ones presented by Codish~\cite{CodishCS14a_The_End_Game} for the last layer because Codish's methods are aimed at the SAT encoding of sorting networks where they consider less input information then what we do. We have also presented an algorithm for finding optimal depth sorting networks that makes use of the newly presented search space size pruning techniques. We have evaluated an implementation of this algorithm to prove (although not novel) the optimal depths of sorting networks for all $n \leq 12$. The runtimes of our program are comparable to those that are reported by Bundala's~\cite{BundalaCCSZ14_Optimal_Depth} state of the art method, although since different machines are used this is to be interpreted as a guideline only. It is important to note that for $n=11$, our pruning techniques reduce the search space size by about $25 \%$ for the fourth level and by $100 \%$ (our pruning technique rejects all networks) for the fifth level when compared to Bundala's approach.

\section{Acknowledgements}


Work supported by the Irish Research Council (IRC).

\bibliographystyle{elsarticle-num}
\bibliography{Algorithm}


\begin{algorithm} [H]
\caption{Pseudo code for our new algorithm for finding minimal depth $n$-input sorting networks. The minimal depth of an $n$-input sorting network is the smallest $d$ for which Exists-Sorting-Network($n, d$) returns true. Proof of correctness of the presented functions is given by comments.}
\label{algo:MinDepthSortNet}
\begin{algorithmic} [1]
\Function{Exists-Sorting-Network}{int $n$, int $d$}
\State $C_0 \leftarrow $ $n$-input comparator network of depth 0
\State $R[0].insert(S_{C_0})$
\Comment $R[i]$ is a set of output sets
\For{$depth = 1, 2, \dots, d $}
	\State $R[depth] \leftarrow$ Generate-Next-Depth$(R[depth - 1], depth, d)$
	\State $R[depth] \leftarrow $ Min-Sets-Up-To-Perm-Refl$(R[depth])$
	\Comment{Algorithm~1 in \cite{Marinov:SortingNetworks:ThirdLevel} for finding the minimal representative itemsets up to permutation and reflection.}
\EndFor

\If{$\lvert R[t] \lvert = 1$}
\Comment{Lemma~\ref{AlgoCorrectExistanceLemma}}
	\If{$\lvert S_{R[t][0]} \lvert = n + 1$}
		\State \Return $true$
		\Comment{Theorem~\ref{SortinNetworkOutputSetSizeTheorem}}
	\EndIf
\EndIf
\State \Return $false$
\EndFunction

\State

\Function{Generate-Next-Depth}{set of output sets $R$, int $n$, int $d$, int $t$}
\State $result \leftarrow \emptyset$
\Comment{set of output sets}
\ForAll{$S_C \in R$}
	\State $Levels \longleftarrow$ Get-All-Levels$ ( S_C, n, d, t ) $
	\ForAll{$L \in$ $Levels$ }
	\Comment{Definition~\ref{ComparatorNetworkDefinition}}
		\State $result \longleftarrow result$ $\bigcup$ $\{ S_{C \oplus L} \}$
	\EndFor
\EndFor
\State \Return $result$
\Comment{$result$ does not contain duplicates}
\EndFunction

\State

\Function{Get-All-Levels}{output set $S_A$, int $n$, int $d$, int $t$}
\State $result \leftarrow \emptyset$
\Comment{set of output sets}
\If{ $ d+4 = t $ }
	\State \Return Get-Fourth-Last-Levels$ ( S_A, n ) $
\EndIf
\If{ $ d+3 = t $ }
	\State \Return Get-Third-Last-Levels$ ( S_A, n ) $
\EndIf
\If{ $ d+2 = t $ }
	\State \Return Get-Second-Last-Levels$ ( S_A, n ) $
\EndIf
\State \Return $G_n$
\Comment{The set of all levels.}
\EndFunction
\end{algorithmic}
\end{algorithm}

\clearpage
\begin{algorithm} [H]
\caption{Pseudo code for constructing the minimal from, to and reach sets for any sorting network $A \oplus B$ by using only the output set $S_A$. Proof of correctness is given by comments referring to theoretical results presented in this paper. The worst case time and space complexities for the function Get-From-To-Reach-Sets are $\BigO{2^n n}$ and $\BigO{2^n}$ respectively.}
\label{algo:construct_reach}
\begin{algorithmic} [1]
\Function{Get-From-To-Reach-Sets}{output set $S_A$}
\For{channel $c = 1 \dots n$}
\State $to[c] \leftarrow \emptyset$
\Comment{$to[c]$ is a set of channels}
\State $from[c] \leftarrow \emptyset$
\Comment{$from[c]$ is a set of channels}
\State $reach[c] \leftarrow \emptyset$
\Comment{$reach[c]$ is a set of channels}
\EndFor

\ForAll{outputs $x,y \in S_A$}
	\If{Get-Different-Coordinates-Count$(x, y) == 1$}
		\State $c \longleftarrow$ Get-First-Different-Coordinate$(x, y)$
		\State $to[ |x| ] \longleftarrow to[ |x| ]$ $\bigcup$ $\{ c \}$
		\State $reach[ |x| ] \longleftarrow reach[ |x| ]$ $\bigcup$ $\{ c \}$
		\State $from[ c ] \longleftarrow from[ c ]$ $\bigcup$ $\{ |x| \}$
		\State $reach[ c ] \longleftarrow reach[ c ]$ $\bigcup$ $\{ |x| \}$
	\EndIf
\EndFor
\State \Return $\langle from, to, reach \rangle$
\EndFunction

\State

\Function{Get-Different-Coordinates-Count}{output $x$, output $y$}
\State $count \longleftarrow 0$
\For{channel $c = 1 \dots n$}
	\If{$x[i] \neq y[i]$}
		\State $count \longleftarrow count + 1$
	\EndIf
\EndFor
\State \Return $count$
\EndFunction

\State

\Function{Get-First-Different-Coordinate}{output $x$, output $y$}
\For{channel $c = 1 \dots n$}
	\If{$x[c] \neq y[c]$}
		\State \Return $c$
	\EndIf
\EndFor
\State \Return $0$
\Comment{Invalid result}
\EndFunction

\end{algorithmic}
\end{algorithm}

\begin{algorithm} [H]
\caption{Pseudo code for finding all network level candidates for the second last level of the comparator network $A$ represented by the output set $S_A$. The worst case time and space complexities for the function Get-Second-Last-Levels are $\BigO{2^n m n}$ and $\BigO{2^n}$ respectively. Correctness of the pseudo code is given by comments referring to results shown in this paper.}
\label{algo:sortable_two}
\begin{algorithmic} [1]
\Function{Get-Second-Last-Levels}{output set $S_A$, int $n$}
\State $levels \longleftarrow \emptyset$
\ForAll{$L \in G_n$}
	\If{Sortable-In-Two-Levels$ ( S_{A \oplus L} ) $ }
		\State \Return $ levels \longleftarrow levels$ $\bigcup$ $\{ L \} $
	\EndIf
\EndFor
\State \Return $levels$
\EndFunction

\State

\Function{Sortable-In-Two-Levels}{output set $S_A$}
\State $\langle from, to, reach \rangle \longleftarrow $ Get-From-To-Reach-Sets$(S_{A \oplus L})$
\For{channel $c = 1 \dots n$}
	\If{$ \lvert from[c] \lvert > 4 $ or $ \lvert to[c] \lvert > 4 $ or $ \lvert reach[c] \lvert > 5 $}
		\State \Return $false$
	\EndIf
\EndFor
\State \Return $true$
\EndFunction
\end{algorithmic}
\end{algorithm}

\begin{algorithm} [H]
\caption{Pseudo code for finding all network level candidates for the third last level of the comparator network $A$ represented by the output set $S_A$. The worst case time and space complexities for the function Get-Third-Last-Levels are $\BigO{2^n m n}$ and $\BigO{2^n}$ respectively. Correctness of the pseudo code is given by comments referring to results shown in this paper.}
\label{algo:sortable_three}
\begin{algorithmic} [1]
\Function{Get-Third-Last-Levels}{output set $S_A$, int $n$}
\State $levels \longleftarrow \emptyset$
\ForAll{$L \in G_n$}
	\If{Sortable-In-Three-Levels$ ( S_{A \oplus L} ) $ }
		\State \Return $ levels \longleftarrow levels$ $\bigcup$ $\{ L \} $
	\EndIf
\EndFor
\State \Return $levels$
\EndFunction

\State

\Function{Sortable-In-Three-Levels}{output set $S_A$}
\State $\langle from, to, reach \rangle \longleftarrow $ Get-From-To-Reach-Sets$(S_{A \oplus L})$
\For{channel $c = 1 \dots n$}
	\If{ $ \lvert from[c] \lvert > 8 $ or $ \lvert to[c] \lvert > 8 $ }
		\State \Return $false$
	\EndIf
\EndFor
\State \Return $true$
\EndFunction
\end{algorithmic}
\end{algorithm}

\begin{algorithm} [H]
\caption{Pseudo code for finding all network level candidates for the fourth last level of the comparator network $A$ represented by the output set $S_A$, also referred to as the ``Look-Ahead''. The worst case time and space complexities for the function Get-Fourth-Last-Levels are $\BigO{2^n m n}$ and $\BigO{2^n}$ respectively. Correctness of the pseudo code is given by comments referring to results shown in this paper.}
\label{algo:sortable_four}
\begin{algorithmic} [1]
\Function{Get-Fourth-Last-Levels}{output set $S_A$, int $n$}
\State $comps \longleftarrow \emptyset$

\ForAll{comparators $ \langle a, b \rangle $}
	\State $L \longleftarrow \{ \langle a, b \rangle \}$
	\State $\langle from, to, reach \rangle \longleftarrow $ Get-From-To-Reach-Sets$(S_{A \oplus L})$
	\If{ $ \lvert from[a] \lvert > 8 $ or $ \lvert from[b] \lvert > 8 $ }
	\Comment{``Look-Ahead''}
		\State $ comps \longleftarrow comps$ $\bigcup$ $\{ \langle a, b \rangle \} $
	\EndIf
\EndFor
\State \Return $\{ L \in G_n ~\lvert~ \forall c \in L ~:~ c \notin comps \}$
\EndFunction
\end{algorithmic}
\end{algorithm}

\end{document}